\documentclass[letterpaper,11pt]{article}

\usepackage[latin1]{inputenc}
\usepackage[english]{babel}
\usepackage{scs}
\usepackage{graphicx}
\usepackage{amssymb,amsmath}
\usepackage{hyperref}   
\usepackage{caption}       
\usepackage{subcaption}     
\usepackage{amsthm}   
\usepackage{color}

\newtheorem{theorem}{Theorem}[section]
\newtheorem{corollary}[theorem]{Corollary}
\newtheorem{lemma}[theorem]{Lemma}

\newtheorem{algorithm}[theorem]{Algorithm}

\newtheoremstyle{normalnoit}{}{}{}{ }{\bf }{.}{ }{}
\theoremstyle{normalnoit}
\newtheorem{example}[theorem]{Example}
\newtheorem{remark}[theorem]{Remark}

\newcommand{\be}{\begin{equation}}
\newcommand{\ee}{\end{equation}}
\newcommand{\bt}{\begin{theorem}}
\newcommand{\et}{\end{theorem}}

\def\invddots{\mathinner{\mskip1mu\raise1pt\vbox{\kern7pt\hbox{.}}\mskip2mu
        \raise4pt\hbox{.}\mskip2mu\raise7pt\hbox{.}\mskip1mu}}
\def\dddots{\mathinner{\mskip1mu\raise4pt\vbox{\kern7pt\hbox{.}}\mskip2mu
        \raise3pt\hbox{.}\mskip2mu\raise1pt\hbox{.}\mskip1mu}}

\begin{document}

%
%
\title{Signal Flow Graph Approach to Efficient DST I-IV Algorithms}
\author{Sirani M. Perera}

\maketitle



\begin{abstract}
\noindent In this paper, fast and efficient discrete sine transformation (DST) algorithms are presented based on the factorization of sparse, scaled orthogonal, rotation, rotation-reflection, and butterfly matrices. These algorithms are completely recursive and solely based on DST I-IV. The presented algorithms have low arithmetic cost compared to the known fast DST algorithms. Furthermore, the language of signal flow graph representation of digital structures is used to describe these efficient and recursive DST algorithms having $(n-1)$ points signal flow graph for DST-I and $n$ points signal flow graphs for DST II-IV. 
\end{abstract}



\section{Introduction}
\label{intro}
Applications of Fast Fourier Transform (FFT) have spread to a very diverse field in applied mathematics and electrical engineering and even the origin of the FFT goes back to analysis of the rotation of Helium molecule \cite{CT65}. By now J. Dongarra and F. Sullivan have categorized FFT as one of the top 10 algorithms of the computer age which had the greatest influence on the development and practice of science and engineering in the 20$^{\rm th}$ century. FFT is used to compute Discrete Fourier Transform (DFT) and its inverse efficiently. On the other hand DFT implementation algorithms employ FFT so FFT and DFT are sometimes used interchangeably. Discrete Sine Transform (DST) is a Fourier-related transform similar to the DFT, but using a purely real matrix. Among applications of the DFT; sine and cosine waves of the DFT with different frequencies are used to classify the traffic monitoring sites into different seasonal patterns \cite{SLZL13}, DST has been identified as the method which generates better results for noise estimation as compared with Discrete Cosine Transform (DCT) and the DFT \cite{DJ11}, discrete fractional sine transform has identified as the method for generating fingerprint templates with high recognition accuracy \cite{Y11}, DCT, DEST, and DFT can be approximated to the Karhunen Loeve Transformation (KLT) and the connection of KLT to the color image compression \cite{BYR06, KGP01, RKH10, RY90}, DST can be used to analyze image reconstruction via signal transition through a square-optical fiber lenses \cite{WSHB00}, spectral interference and additive wideband noise on the accuracy of the normalized frequency estimator can be investigated through discrete-time sine-wave \cite{BDP12}, to mention a few. Together with the above, the engagement of DCT and DST in image processing, signal processing, finger print enhancement, quick response code (QR code), and multimode interface can also be seen in e.g., \cite{B13, CR12, FMW13, HSMR12, JKJ09, KSS14, KSN14, KS78, KR09, LKKP13, MPH12,  L10, LC09,  S86, S99, VL92, VZR12, VP09}.          
\newline\newline
The family of DFT consists of eight versions (I-VIII) of DCT and DST and these versions appear depending on odd or even type and also with respect to different Neumann and Dirichlet boundary conditions \cite{BYR06, M94, RS11, S99}. Though there are eight versions, depending on applications in transform coding and digital filtering of signals, we consider DCT and DST matrices as varying from I to IV types. Let us consider four orthogonal types of DST having superscripts to denote the type of DST and a subscript to denote the order of DST in the matrix form;    
\begin{equation}
\label{tbl:sct}
\begin{array}{cc}
DST-I: & S_{n-1}^I  =  \sqrt{\frac{2}{n}}\left [\:{\rm sin}\:\frac{(j+1)(k+1)\pi }{n}  \right ]_{j,\:k\: = 0}^{n-2}, 
\\
DST-II: & S_{n}^{II}  = \sqrt{\frac{2}{n}}\left [ \epsilon_n(j+1)\:{\rm sin}\:\frac{(j+1)(2k+1)\pi }{2n}  \right ]_{j,\:k\:=0}^{n-1} \\
DST-III: & S_{n}^{III}= \sqrt{\frac{2}{n}}\left [ \epsilon_n(k+1)\:{\rm sin}\:\frac{(2j+1)(k+1)\pi }{2n}  \right ]_{j,\:k\:=0}^{n-1},
\\
DST-IV: & S_{n}^{IV} = \sqrt{\frac{2}{n}}\left [ \:{\rm sin}\:\frac{(2j+1)(2k+1)\pi }{4n}  \right ]_{j,\:k\:=0}^{n-1}
\end{array}
\end{equation}
where $\epsilon_n(0)=\epsilon_n(n)=\frac{1}{\sqrt{2}}$, $\epsilon_n(j)=1$ for $j \in \{1,2,\cdots,n-1\}$ and $n \geq 2$ is an even integer. Among DST I-IV transformations, $S_{n-1}^I$ and $S_n^{IV}$ were introduced in \cite{J76, J79} and $S_{n}^{II}$ and its inverse $S_{n}^{III}$ were introduced in \cite{KS78} into digital signal processing. DST-II is a complementary or alternative transform to DCT-II which is used in transform coding. Like DFT and DCT, these DST matrices hold linearity, convolution-multiplication, and shift properties. 
\newline\newline
Among different mathematical techniques used to derive fast algorithms for discrete cosine and sine transformations, the polynomial arithmetic technique (see e.g. \cite{PM03,ST91}) and the matrix factorization technique (see e.g. \cite{ BYR06, PT05, RY80, SO13, S14, W84}) can be seen as the dominant techniques. Apart from these two main techniques some other authors (see e.g. \cite{KO96, OOW03}) have used different techniques like displacement approach and polynomial division in matrix form to derive factorizations for DCT and DST.  Efficient algorithms for DCT or DST of radix-2 length $n$ require about $2\:n \:log_2 n$ flops. Such a DCT or DST algorithm generates a factorization of these matrices having sparse and non-orthogonal matrices. Thus, if the factorization for DCT or DST does not preserve orthogonality the resulting DCT or DST algorithms lead to inferior numerical stability (see e.g. \cite{TZ00}). The matrix factorization for DST I in \cite{RY80} used the results in \cite{CSF77} to decompose DST I into DCT and DST. Also the decomposition for DCT II in \cite{W84} is a slightly different version of the result in \cite{CSF77}. Though one can find orthogonal matrix factorization for DCT and DST in \cite{W84}, the resulting algorithms in \cite{W84} are not completely recursive and hence do not lead to simple recursive algorithms. An alternative factorization for DCT I-IV in \cite{S14} and DST I-IV in \cite{SO13} can be seen in \cite{BYR06, PT05} but the factorizations in the latter papers are not solely dependent on DCT I-IV or DST I-IV. Moreover \cite{BYR06} has used the same factorization for DST-II and DST-IV as in \cite{W84}. However one can use these \cite{BYR06, W84} results to derive recursive, stable and radix-2 algorithms as stated in \cite{PT05, SO13, S14}. 
\newline\newline
In electrical engineering, control theory, system engineering, theoretical computer science, etc. signal flow graphs represent realizations of systems as electronic devices. The objective here is to build a device to implement or realize algebraic operations used in sparse and orthogonal factorization of fast and recursive DST I-IV algorithms. Based on the factorization of DFT, DCT, and DST matrices one can design signal flow graphs such as: 8-point signal flow graphs on various fast DCT and DST algorithms having sparse and/or orthogonal factorization in \cite{BYR06}, signal flow graphs for forward and backward modified DCT implementations with $n=12$ and also with mix-radix decomposition of $n=12$ in \cite{B11}, fast DST-VII and DCT-II algorithms based signal flow graphs for $2n+1$ points and $n(2n+1)$ points DCT-II in \cite{R13}, signal flow graphs representation of the direct 2-D DCT-II and 2-D DST-II computation and their inverses for $16\times 16, 8 \times 16, 4 \times 16, 16 \times  8$, and $16 \times  4$ block sizes in \cite{BR00}, signal flow graphs of the coordinate rotation digital computer-based $n$ points DCT-II, DCT-III and DST-II, DST-III in \cite{HXL14}, signal flow graphs based Jacob rotation for $n/2$ points DCT-IV and modified DCT-IV in \cite{B09}, and signal flow graphs based on hybrid jacket-Hadamard matrix for $n$ points DCT-II, DST-II, and DFT-II  in \cite{LKKP13}. However, there is no paper on signal flow graphs based on fast and completely recursive DST I-IV algorithms having sparse, scaled orthogonal factorization, rotation, rotation-reflection matrix factorizations and especially the generalization of $n$ points signal flow graphs covering all DST matrices of types I to IV. Hence in this paper we modify the sparse and orthogonal factorizations of stable DST I-IV  algorithms proposed in \cite{SO131, SO13} to derive fast (compared to known algorithms), efficient, and completely recursive sole algorithms based on DST I-IV having sparse, scaled orthogonal, rotation, rotation-reflections matrices and to discuss the arithmetic complexity of these fast DST I-IV algorithms. Furthermore, the paper presents generalized $n-1$ points signal flow graph for DST-I and $n$ points signal flow graphs for DST II-IV based on the recursive DST I-IV algorithms.     
\newline\newline
In section 2 we modify the factorizations derived in \cite{SO13} to express fast, efficient, and completely recursive sole algorithms for DST I-IV having scaled orthogonal, sparse, rotation, rotation-reflection, and butterfly matrices. Next, in section 3, we derive a number of additions and multiplications required to compute these fast and efficient DST I-IV algorithms and illustrate the numerical results based on that. In section 4, we develop and then generalize signal flow graphs for $n-1$ points DST-I algorithm and $n$ points DST II-IV algorithms.   
\section{Efficient and recursive DST algorithms having sparse, scaled orthogonal, and rotational-reflection factors}
\label{sec:factor}
This section presents fast, efficient, and completely recursive DST algorithms solely defined via DST I-IV having sparse, scaled orthogonal, rotational, and rotational-reflection factors by modifying radix-2, recursive, and stable DST I-IV algorithms having sparse and orthogonal factors introduced in \cite{SO13}. The purpose of this is to significantly reduce the number of multiplications required to compute DST I-IV algorithms compared to the known fast, efficient, and stable DST algorithms having sparse factorizations.  
\\\\
By applying the permutation matrix to each sine transform matrix and using the trigonometric addition, complementary, and supplementary identities, one can derive the orthogonal matrix factorization for DST I-IV as in \cite{SO13}.  
\newline\newline
In the following we state the collection of sparse and orthogonal matrices which are frequently used in this paper. For a given vector ${\bf x} \in \mathbb{R}^n$, let us introduce an involution matrix $\tilde{I}_n$ by 
\begin{equation}
\tilde{I}_n\:\textbf{x}= \begin{array}{c}
\left[x_{n-1},x_{n-2},\cdots,x_0 \right ]^T 
\end{array},
\nonumber
\end{equation}
a diagonal matrix $D_n$ by
\begin{equation}
 D_n \: \textbf{x}  =\left \{ \begin{array}{cc}
\left[x_0, -x_1, x_2, -x_3,\cdots, x_{n-1}, -x_{n-1} \right ]^T & \textrm{even\:}n, \\
\left[x_0, -x_1, x_2, -x_3,\cdots, -x_{n-1}, x_{n-1} \right ]^T & \textrm{odd\:}n
\end{array} 
\right.
\nonumber
\end{equation}
and, for $n \geq 3$ an even-odd permutation matrix $P_n$ by
\begin{equation}
P_n  \: \textbf{x} =\left \{ \begin{array}{cc}
\left[x_0,x_2,\cdots,x_{n-2},x_1,x_3,\cdots,x_{n-1} \right ]^T & \textrm{even\:}n, \\
\left[x_0,x_2,\cdots,x_{n-1},x_1,x_3,\cdots,x_{n-2} \right ]^T & \textrm{odd\:}n.
\end{array} \right.
\nonumber
\end{equation}
For even integer $n \geq 4$, we introduce sparse and orthogonal matrices: 
\[
\widehat{H}_{n-1}=\frac{1}{\sqrt{2}}\left[\begin{array}{rcr}
 I_{\frac{n}{2}-1}&  & \widetilde{I}_{\frac{n}{2}-1}\\ 
 & \sqrt{2} & 
\\
I_{\frac{n}{2}-1} &  &-\widetilde{I}_{\frac{n}{2}-1} \\ 
\end{array}\right],
\hspace{.2in}
{H}_{n}=\frac{1}{\sqrt{2}}\left[\begin{array}{rcr}
 I_{\frac{n}{2}}&  & \widetilde{I}_{\frac{n}{2}}\\ 
\\
 I_{\frac{n}{2}} &  &-\widetilde{I}_{\frac{n}{2}} \\ 
\end{array}\right],
\]
\[
V_n=\left[\begin{array}{ccc}
1 &  & \\ 
 & \frac{1}{\sqrt{2}}\left[\begin{array}{rr}
I_{\frac{n}{2}-1} & -I_{\frac{n}{2}-1}\\ 
 -I_{\frac{n}{2}-1}& -I_{\frac{n}{2}-1}
\end{array}\right] & \\ 
 &  & -1
\end{array}\right]
\left[\begin{array}{cc}
 \widetilde{I}_{\frac{n}{2}}& \\ 
 & D_{\frac{n}{2}}
\end{array}\right],
\]
and also a rotational-reflection matrix:
\[
\begin{aligned}
{Q}_n & =\begin{bmatrix}
D_{\frac{n}{2}} & \\ 
 & I_{\frac{n}{2}}
\end{bmatrix}\begin{bmatrix}
{\rm diag}\:S_{\frac{n}{2}} &\left ( {\rm diag}\:C_{\frac{n}{2}} \right ) \widetilde{I}_{\frac{n}{2}}\\ 
-\widetilde{I}_{\frac{n}{2}}\left ( {\rm diag}\:C_{\frac{n}{2}} \right ) &  {\rm diag}\:\left ( \widetilde{I}_{\frac{n}{2}}S_{\frac{n}{2}}  \right )\end{bmatrix}
\\ 
& = \small \begin{bmatrix}
{\rm sin}\:\frac{\pi}{4n} &   & & & &  &  & {\rm cos}\:\frac{\pi}{4n}\\ 
 & -{\rm sin}\:\frac{3\pi}{4n} &  &  & &  & -{\rm cos}\:\frac{3\pi}{4n} & \\ 
&    & \ddots& & &\invddots &    & \\
 &  &  &  & -{\rm sin}\:\frac{(n-1)\pi}{4n}&-{\rm cos}\:\frac{(n-1)\pi}{4n} & &  &  & \\ 
 &  &  & &-{\rm cos}\:\frac{(n-1)\pi}{4n} &{\rm sin}\:\frac{(n-1)\pi}{4n} & &  &  & \\ 
&    &  \invddots  & & &\ddots &    & \\
  & -{\rm cos}\:\frac{3\pi}{4n}&  &  & &  & {\rm sin}\:\frac{3\pi}{4n} & \\ 
-{\rm cos}\:\frac{\pi}{4n} &  &  &   & &  &  & {\rm sin}\:\frac{\pi}{4n}
\end{bmatrix},
\end{aligned}
\]
where 
\[
C_{\frac{n}{2}}=\left [ {\rm cos} \frac{(2k+1)\pi}{4n} \right ]_{k=0}^{\frac{n}{2}-1}\hspace{.2in}{\rm and}\hspace{.2in}
S_{\frac{n}{2}}=\left [ {\rm sin} \frac{(2k+1)\pi}{4n} \right ]_{k=0}^{\frac{n}{2}-1}.
\]
\subsection{Stable, recursive, radix-2 DST I-IV algorithms having sparse and orthogonal factors}
\label{subs:Salgo}
Before developing DST matrix factorization based on fast, efficient, and completely recursive DST I-IV algorithms, let us state stable, simple, recursive, radix-2 DST I-IV algorithms having sparse and orthogonal factors derived in \cite{SO13}. 
\\\\
The algorithms for DST I-IV are stated in order $S_n^{II}$, $S_n^{IV}$, $S_n^{III}$ and $S_{n-1}^{I}$ respectively. 
\\\\
In \cite{SO13}, orthogonal factorizations for types II and IV of discrete sine transform matrices are given by
\begin{center}
$
S_n^{II}=P_n^T
\begin{bmatrix}
S_{\frac{n}{2}}^{IV} & 0 \\
0& S_{\frac{n}{2}}^{II}
\end{bmatrix} H_n
\hspace{.7in}
{\rm and} 
\hspace{.7in}
S_n^{IV}=P_n^TV_n
\begin{bmatrix}
S_{\frac{n}{2}}^{II} & 0 \\
0& S_{\frac{n}{2}}^{II}
\end{bmatrix} Q_n
$
\end{center}
Thus, the recursive algorithms for DST-II and DST-IV can be stated via algorithms $\bf{(\ref{algo:s2})}$ and $\bf{(\ref{algo:s4})}$ respectively.  
\begin{algorithm} ${\bf sin2}( n)$ \\\\
\label{algo:s2}
Input: $n=2^t(t \geq 1)$, $n_1=\frac{n}{2}$. 
\begin{enumerate}
\item{If $n=2$, then} \\
$
{ S2}:=\frac{1}{\sqrt{2}}\left[\begin{array}{rr}
1 & 1\\
1 & -1
\end{array}\right].
$
\item{ If $n \geq 4$, then}\\
$
\begin{array}{c}
\begin{aligned}
\hspace{.1in} { M1}  :=&\: {\bf sin4} \left( n_1 \right),   \\
\hspace{.1in} { M2}  :=&\: {\bf sin2} \left(n_1 \right),  \\
\hspace{.1in} { S2} := &\: P_n^T\left({\rm blkdiag}(M1, M2) \right) H_n.   
\end{aligned}
\end{array}
$
\end{enumerate}
Output: ${ S2}=S_n^{II}$.
\end{algorithm}
\begin{algorithm} ${\bf sin4}(n)$ \\\\
\label{algo:s4}
Input: $n=2^t(t \geq 1)$, $n_1=\frac{n}{2}$. 
\begin{enumerate}
\item{If $n=2$, then} \\
$
{S4}:=\left[\begin{array}{rr}
\sin \frac{\pi}{8} & \cos \frac{\pi}{8}\\
\cos \frac{\pi}{8} & -\sin \frac{\pi}{8}
\end{array}\right].
$
\item{ If $n \geq 4$, then}\\
$
\begin{array}{c}
\begin{aligned}
\hspace{.1in} { M1} :=&\: {\bf sin2} \left( n_1 \right),   \\
\hspace{.1in} { M2} :=&\: {\bf sin2} \left( n_1 \right), \\ 
\hspace{.1in} { L} :=&\: V_n  \left({\rm blkdiag} (M1, M2) \right)Q_n, \\ 
\hspace{.1in} { S4} :=&\:  P_n^T { L}.   
\end{aligned}
\end{array}
$
\end{enumerate}
Output: ${ S4}=S_n^{IV}$.
\end{algorithm}
The transpose of DST-II is DST-III. Thus DST-III algorithm can be computed via algorithm $\bf{(\ref{algo:s2})}$. Observe that this algorithm executes recursively with DST-II and DST-IV algorithms.  
\begin{algorithm} ${\bf sin3}( n)$ \\\\
\label{algo:s3}
Input: $n=2^t(t \geq 1)$, $n_1=\frac{n}{2}$. 
\begin{enumerate}
\item{If $n=2$, then} \\
$
{S3}:=\frac{1}{\sqrt{2}}\left[\begin{array}{rr}
1 & 1\\
1 & -1
\end{array}\right].
$
\item{ If $n \geq 4$, then}\\
$
\begin{array}{c}
\begin{aligned}
\hspace{.1in} {M1}  :=&\: {\bf sin4} \left(n_1 \right),   \\
\hspace{.1in} {M2}  :=&\: {\bf sin3} \left( n_1 \right),  \\
\hspace{.1in} { S3} :=&\:  H_n^T \left({\rm blkdiag} (M1, M2)\right) P_n. \\
\end{aligned}
\end{array}
$
\end{enumerate}
Output: ${ S3}=S_n^{III}$.
\end{algorithm}
\noindent
\\\\
Following \cite{SO13}, orthogonal factorizations for type I discrete sine transform matrix is given by
\begin{center}
$
S_{n-1}^{I}=P_{n-1}^T
\begin{bmatrix}
C_{\frac{n}{2}}^{III} & 0 \\
0& C_{\frac{n}{2}-1}^{I}
\end{bmatrix} \widehat{H}_{n-1}
$
\end{center}
Thus, the recursive algorithm for DST-I can be stated via algorithm $\bf{(\ref{algo:s1})}$. Note that this algorithm runs recursively with DST II-IV algorithms.   
\begin{algorithm} ${\bf sin1}( n-1)$ \\\\
\label{algo:s1}
Input: $n=2^t(t \geq 1)$, $n_1=\frac{n}{2}$. 
\begin{enumerate}
\item{If $n=2$, then} \\
$
{ S1}:= 1.
$
\item{ If $n \geq 4$, then} \\
$
\begin{array}{c}
\begin{aligned}
\hspace{.1in} { M1}  :=&\: {\bf sin3} \left(n_1 \right),   \\
\hspace{.1in} { M2}  :=&\: {\bf sin1} \left( n_1-1 \right),  \\
\hspace{.1in} { S1} :=&\:  P_{n-1}^T \left({\rm blkdiag}(M1, M2) \right) \widehat{H}_{n-1}. \\
\end{aligned}
\end{array}
$
\end{enumerate}
Output: ${ S1}=S_{n}^{I}$.
\end{algorithm}
\subsection{Efficient and completely recursive DST I-IV algorithms}
\label{subs:FCalgo}
In this section, we present fast, efficient, and completely recursive DST I-IV (say NDST I-IV) algorithms using DST I-IV algorithms stated via \ref{algo:s2}, \ref{algo:s4}, \ref{algo:s3}, and \ref{algo:s1} i.e. we introduce DST I-IV algorithms having sparse, scaled orthogonal, rotational, rotational-reflection factors so that DST I-IV are orthogonal w. r. t. the scale factor $\frac{1}{\sqrt{n}}$. In order to reduce number of multiplications, we move the factor $\frac{1}{\sqrt{2}}$ in ${H}_{n}, \widehat{H}_{n-1},$ and $V_n$ without changing the rotation-reflection matrix $Q_n$ so that we compute $\sqrt{n}\:S_n^{II}, \sqrt{n}\:S_n^{IV}, \sqrt{n}\:S_n^{III}$, and $\sqrt{n}\:S_{n-1}^{I}$ respectively. Let us state the corresponding new algorithms via ${\bf nsin2}(n)$, ${\bf nsin4}(n)$, ${\bf nsin3}(n)$, and ${\bf nsin1}(n-1)$ respectively.  
\begin{algorithm} ${\bf nsin2}( n)$ \\\\
\label{algo:ms2}
Input: $n=2^t(t \geq 1)$, $n_1=\frac{n}{2}$. 
\begin{enumerate}
\item{If $n=2$, then} \\
$
{ MS2}:=\left[\begin{array}{rr}
1 & 1\\
1 & -1
\end{array}\right].
$
\item{ If $n \geq 4$, then}\\
$
\begin{array}{c}
\begin{aligned}
\hspace{.1in} { M1}  :=&\: {\bf nsin4} \left( n_1 \right),   \\
\hspace{.1in} { M2}  :=&\: {\bf nsin2} \left(n_1 \right),  \\
\hspace{.1in} { MS2} := &\: P_n^T\left({\rm blkdiag}(M1, M2) \right) \left( \sqrt{2}\: H_n \right).   
\end{aligned}
\end{array}
$
\end{enumerate}
Output: ${ MS2}=\sqrt{n}\:S_n^{II}$.
\end{algorithm}
\begin{algorithm} ${\bf nsin4}(n)$ \\\\
\label{algo:ms4}
Input: $n=2^t(t \geq 1)$, $n_1=\frac{n}{2}$. 
\begin{enumerate}
\item{If $n=2$, then} \\
$
{MS4}:=\sqrt{2}\:\left[\begin{array}{rr}
\sin \frac{\pi}{8} & \cos \frac{\pi}{8}\\
\cos \frac{\pi}{8} & -\sin \frac{\pi}{8}
\end{array}\right].
$
\item{ If $n \geq 4$, then}\\
$
\begin{array}{c}
\begin{aligned}
\hspace{.1in} { M1} :=&\: {\bf nsin2} \left( n_1 \right),   \\
\hspace{.1in} { M2} :=&\: {\bf nsin2} \left( n_1 \right), \\ 
\hspace{.1in} { L} :=&\: \left( \sqrt{2} V_n \right)  \left({\rm blkdiag} (M1, M2) \right)Q_n, \\ 
\hspace{.1in} { MS4} :=&\:  P_n^T { L}.   
\end{aligned}
\end{array}
$
\end{enumerate}
Output: ${ MS4}=\sqrt{n}\:S_n^{IV}$.
\end{algorithm}
The fast, efficient, and completely recursive DST-III algorithm can be computed using the DST-II so that it runs recursively with ${\bf nsin2}(n)$ and ${\bf nsin4}(n)$ algorithms.
\begin{algorithm} ${\bf nsin3}( n)$ \\\\
\label{algo:ms3}
Input: $n=2^t(t \geq 1)$, $n_1=\frac{n}{2}$. 
\begin{enumerate}
\item{If $n=2$, then} \\
$
{MS3}:=\left[\begin{array}{rr}
1 & 1\\
1 & -1
\end{array}\right].
$
\item{ If $n \geq 4$, then}\\
$
\begin{array}{c}
\begin{aligned}
\hspace{.1in} {M1}  :=&\: {\bf nsin4} \left(n_1 \right),   \\
\hspace{.1in} {M2}  :=&\: {\bf nsin3} \left( n_1 \right),  \\
\hspace{.1in} { MS3} :=&\:  \left( \sqrt{2} \:H_n^T \right) \left({\rm blkdiag} (M1, M2)\right) P_n. \\
\end{aligned}
\end{array}
$
\end{enumerate}
Output: ${ MS3}=\sqrt{n}\:S_n^{III}$.
\end{algorithm}
Finally, the fast, efficient, and completely recursive DST I algorithm can be stated as follows. Note that this algorithm runs recursively with ${\bf nsin2}(n)$ , ${\bf nsin4}(n)$, and ${\bf nsin3}(n)$ algorithms.
\begin{algorithm} ${\bf nsin1}( n-1)$ \\\\
\label{algo:ms1}
Input: $n=2^t(t \geq 1)$, $n_1=\frac{n}{2}$. 
\begin{enumerate}
\item{If $n=2$, then} \\
$
{ MS1}:= 1.
$
\item{ If $n \geq 4$, then} \\
$
\begin{array}{c}
\begin{aligned}
\hspace{.1in} { M1}  :=&\: {\bf nsin3} \left(n_1 \right),   \\
\hspace{.1in} { M2}  :=&\: {\bf nsin1} \left( n_1-1 \right),  \\
\hspace{.1in} { MS1} :=&\:  P_{n-1}^T \left({\rm blkdiag}(M1, M2) \right) \left( \sqrt{2} \widehat{H}_{n-1} \right). \\
\end{aligned}
\end{array}
$
\end{enumerate}
Output: ${ MS1}=\sqrt{n}\:S_{n-1}^{I}$.
\end{algorithm}

\subsection{Examples for computing efficient and completely recursive DST I-IV algorithms}
\label{ex:CnS}
Here we state examples for computing fast, efficient, and recursive DST I-IV algorithms having sparse, scaled orthogonal, rotational, and rotational-reflection matrix factorizations based on DST I-IV algorithms ${\bf nsin1}(n-1)$, ${\bf nsin2}(n)$, ${\bf nsin3}( n)$, and ${\bf nsin4}( n)$ for $n=8$. Later in section \ref{sec:SFG}, we use the factorizations for DST I-IV matrices to develop and generalize $n$ points signal flow graphs for DST I-IV algorithms.


\begin{example} \label{Es32} By following algorithms $\bf{(\ref{algo:ms1})}$, $\bf{(\ref{algo:ms2})}$, $\bf{(\ref{algo:ms3})}$, and $\bf{(\ref{algo:ms4})}$, the factorization for DST-I given by:

\begin{equation}
\label{exs19}
\begin{aligned}
&\sqrt{8}\: S_7^{I}\\
&={P}_7^T\begin{bmatrix}
\sqrt{2}\:{H}_{4}^T & 0\\ 
0 & {P}_3^T
\end{bmatrix}\begin{bmatrix}
\sqrt{2}\:S_2^{IV} &0  & 0 & 0\\ 
 0& \sqrt{2}\:S_2^{III} &  0&0 \\ 
 0& 0 & \sqrt{2}\:S_2^{III} & 0\\ 
0 & 0 & 0 & \sqrt{2}\:S_1^{I}
\end{bmatrix}\begin{bmatrix}
{P}_4 & 0\\ 
0 & \sqrt{2}\:\widehat{H}_{3}
\end{bmatrix}\sqrt{2}\:\widehat{H}_7
\end{aligned}
\end{equation}
where
\[
{P}_3=\begin{bmatrix}
1 & 0 & 0 \\ 
0 & 0 &1  \\ 
0 & 1 & 0 \\ 
\end{bmatrix},\:S_1^I=1,\:\sqrt{2}\:S_2^{III}=\left[\begin{array}{rr}
1 &1 \\ 
1 &-1 
\end{array}\right],\:\sqrt{2}\:S_2^{IV}=\sqrt{2}\:\left[\begin{array}{rr}
{\rm sin}\frac{\pi}{8} &{\rm cos}\frac{\pi}{8} \\ 
{\rm cos}\frac{\pi}{8} &-{\rm sin}\frac{\pi}{8}
\end{array}\right]
\]
\[
\sqrt{2}\:\widehat{H}_3=\left[\begin{array}{rrr}
1 &0  & 1\\ 
0& \sqrt{2} & 0
\\
 1&  0& -1
\end{array}\right]
,\:
\sqrt{2}\:\widehat{H}_7=\left[\begin{array}{rrr}
I_2 &  0& \tilde{I}_2 \\ 
0  & \sqrt{2} & 0  \\ 
 I_2 & 0  & -\tilde{I}_2 
\end{array}\right]
\]
\end{example}

\begin{example} \label{Es34}
By following algorithms $\bf{(\ref{algo:ms2})}$ and $\bf{(\ref{algo:ms4})}$, the factorization for DST-II given by:
\begin{equation}
\label{exs28}
\begin{aligned}
& \sqrt{8}\:S_8^{II}\\
&={P}_8^T\begin{bmatrix}
{P}_{4}^T & 0\\ 
0 & {P}_4^T
\end{bmatrix}\:\begin{bmatrix}
\sqrt{2}\:V_{4} & 0\\ 
0 & I_4
\end{bmatrix}\begin{bmatrix}
\sqrt{2}\:S_2^{II} &0  &0  & 0\\ 
0 & \sqrt{2}\:S_2^{II} & 0 &0 \\ 
 0& 0 & \sqrt{2}\:S_2^{IV} & 0\\ 
0 & 0 & 0 & \sqrt{2}\:S_2^{II}
\end{bmatrix}\\
& \hspace{.3in} \begin{bmatrix}
{Q}_{4} & 0\\ 
0 & \sqrt{2}\:H_4
\end{bmatrix}\sqrt{2}\:{H}_8
\end{aligned}
\end{equation}
where 
\[
\sqrt{2}\:V_4=\left[\begin{array}{rrrr}
0 &\sqrt{2} &0 & 0\\ 
1 & 0 & -1& 0\\
-1 & 0& -1& 0\\
 0& 0&0 &\sqrt{2}\\ 
\end{array}\right],\:Q_4=\begin{bmatrix}
{\rm sin}\frac{\pi}{16} & 0&0 &{\rm cos}\frac{\pi}{16}\\ 
0 &-{\rm sin}\frac{3\pi}{16} & -{\rm cos}\frac{3\pi}{16}&0\\
0 &-{\rm cos}\frac{3\pi}{16} & {\rm sin}\frac{3\pi}{16}&0\\
-{\rm cos}\frac{\pi}{16} &0 & 0&{\rm sin}\frac{\pi}{16}\\ 
\end{bmatrix}\]
\end{example}

\begin{example} \label{Es36}
By following algorithms $\bf{(\ref{algo:ms2})}$, $\bf{(\ref{algo:ms4})}$ and $\bf{(\ref{algo:ms3})}$, the factorization for DST-III given by:
\begin{equation}
\label{exs38}
\begin{aligned}
&\sqrt{8}\:S_8^{III}\\
&=\sqrt{2}\:{H}_8^T\begin{bmatrix}
{P}_4^T & 0\\ 
0 & I_4
\end{bmatrix}\:\begin{bmatrix}
\sqrt{2}\:V_{4} & 0\\ 
0 & \sqrt{2}\:H_4^T
\end{bmatrix}\begin{bmatrix}
\sqrt{2}\:S_2^{II} &  0&0  & 0\\ 
 0& \sqrt{2}\:S_2^{II} & 0 &0 \\ 
 0&0  & \sqrt{2}\:S_2^{IV} & 0\\ 
0 &  0&0  & \sqrt{2}\:S_2^{III}
\end{bmatrix}\\
& \hspace{.3in} \begin{bmatrix}
{Q}_{4} & 0\\ 
0 & P_4
\end{bmatrix}{P}_8
\end{aligned}
\end{equation}
\end{example}

\begin{example} \label{Es38}
By following algorithms $\bf{(\ref{algo:ms2})}$ and $\bf{(\ref{algo:ms4})}$, the factorization for DST-IV given by:
\begin{equation}
\begin{aligned}
& \sqrt{8}\:S_8^{IV}\\
&=P_8^T\:\sqrt{2}V_8\begin{bmatrix}
{P}_4^T & 0\\ 
0 & {P}_4^T
\end{bmatrix}\begin{bmatrix}
\sqrt{2}\:S_2^{IV} & 0 &0  & 0\\ 
0 & \sqrt{2}\:S_2^{II} &  0&0 \\ 
 0& 0 & \sqrt{2}\:S_2^{IV} & 0\\ 
0 & 0 &  0& \sqrt{2}\:S_2^{II}
\end{bmatrix}\\
& \hspace{.3in}
\begin{bmatrix}
\sqrt{2}\:{H}_{4} & 0\\ 
0 & \sqrt{2}\:{H}_{4}
\end{bmatrix}Q_8
\end{aligned}
\label{exs48}
\end{equation}
where 
\[
\sqrt{2}V_8=\left[\begin{array}{rrrrrrrr}
0 &0  & 0 & \sqrt{2} & 0 &0  & 0 & 0\\ 
0 & 0 &1  & 0 & -1 & 0 & 0 & 0\\ 
 0 & 1 & 0 & 0 &0  &  1&  0& 0\\ 
1 & 0 & 0 &0  &0  &0  & -1 & 0\\ 
 0& 0 & -1 & 0 & -1 &0  & 0 & 0\\ 
 0& -1 & 0 & 0 & 0 &  1& 0 & 0\\ 
 -1 & 0 & 0 & 0 &0  & 0&  -1& 0\\ 
 0 & 0 & 0 & 0 & 0 & 0 & 0 & \sqrt{2}
\end{array}\right]
\]
\[
Q_8=\arraycolsep=.9pt\begin{bmatrix}
{\rm sin}\:\frac{\pi}{32} &  0&  0& 0 & 0 & 0 &0  & {\rm cos}\:\frac{\pi}{32}\\ 
0 & -{\rm sin}\:\frac{3\pi}{32} & 0 & 0 & 0 &  0& -{\rm cos}\:\frac{3\pi}{32} & \\ 
0 &0  &  {\rm sin}\:\frac{5\pi}{32}&  0& 0 &{\rm cos}\:\frac{5\pi}{32}  &  0& 0\\ 
0 & 0 &0  &  -{\rm sin}\:\frac{7\pi}{32}& -{\rm cos}\:\frac{7\pi}{32} & 0 &  0& 0\\ 
 0 &0  & 0 & -{\rm cos}\:\frac{7\pi}{32} & {\rm sin}\:\frac{7\pi}{32} & 0 & 0 & 0 \\ 
 0 & 0 & -{\rm cos}\:\frac{5\pi}{32} & 0 &  0&{\rm sin}\:\frac{5\pi}{32}  &  0 & 0 \\ 
 0 & -{\rm cos}\:\frac{3\pi}{32}& 0 & 0 &0  & 0 & {\rm sin}\:\frac{3\pi}{32} & 0\\ 
-{\rm cos}\:\frac{\pi}{32} & 0 & 0 & 0 & 0 & 0 &  0& {\rm sin}\:\frac{\pi}{32}
\end{bmatrix}
\]
\end{example}

\section{Arithmetic complexity of computing fast, efficient, and completely recursive DST I-IV algorithms having sparse and scaled orthogonal factors}
\label{sec:cost}
The number of additions and multiplications required to compute DST I-IV algorithms via ${\bf nsin1}(n-1)$, ${\bf nsin2}(n)$, ${\bf nsin3}(n)$, ${\bf nsin4}( n)$ are considered in this section. The number of additions and multiplications required to compute, say length $n$,  DST-II algorithm (${\bf nsin2}(n)$) are denoted by $\#a(\textrm{NDST-II}, n)$ and $\#m(\textrm{NDST-II}, n)$ respectively. Note that the multiplication of $\pm 1$ and permutations are not counted. At the end of the section we illustrate numerical results based on the number of additions and multiplication required to compute these DST I-IV algorithms.
\subsection{Arithmetic complexity of DST I-IV algorithms}
\label{subs:Scost}
Here we address the arithmetic cost of computing fast, efficient, and recursive DST I-IV algorithms having sparse, scaled orthogonal, rotational, and rotational-reflection factors. The complexity of computing these DST I-IV algorithms are expressed first by calculating the arithmetic complexity of DST-II algorithm and then using it to compute the complexity of DST-IV, DST-III and DST-I algorithms respectively. 
\begin{lemma}\label{Ls41}
Let $n=2^t\:(t \geq 2)$ be given. If DST-II algorithm (${\bf nsin2}(n)$) is computed using algorithms $\bf{(\ref{algo:ms2})}$ and $\bf{(\ref{algo:ms4})}$ then the arithmetic cost of computing length $n$ DST-II algorithm is given by  
\begin{eqnarray}
\#a(\textrm{NDST-II}, n) &=& \frac{4}{3}nt-\frac{8}{9}n-\frac{1}{9}(-1)^t+1,
\nonumber \\
\#m(\textrm{NDST-II}, n) &=& \frac{2}{3}nt+\frac{2}{9}n+\frac{7}{9}(-1)^t-1.
\label{cs2s4}
\end{eqnarray}
\end{lemma}
\begin{proof}
From algorithms $\bf{(\ref{algo:ms2})}$ and $\bf{(\ref{algo:ms4})}$
\begin{eqnarray}
\#a(\textrm{NDST-II}, n) & =& \#a\left(\textrm{NDST-II}, \frac{n}{2} \right) + \#a\left(\textrm{NDST-IV}, \frac{n}{2} \right) + \#a\left(\sqrt{2}\:{H}_n \right)
\nonumber \\
\#a\:(\textrm{NDST-IV}, n) & =& \#a\left(\sqrt{2}\:V_n \right) + 2\cdot \#a\left(\textrm{NDST-II}, \frac{n}{2} \right) + \#a\left(Q_n \right)
\label{as24}
\end{eqnarray}
Referring the structures of ${H}_n$, $V_n$, and $Q_n$ 
\begin{eqnarray}
\#a \left(\sqrt{2}\:{H}_n \right) = n, \:\: \#m \left(\sqrt{2}\:{H}_n \right) = 0
\nonumber \\
\#a \left(\sqrt{2}\:V_n \right) = n-2, \:\: \#m \left(\sqrt{2}\:V_n \right) = 2
\nonumber \\
\#a \left(Q_n \right) = n, \:\:  \#m \left(Q_n \right) = 2n
\label{rus24}
\end{eqnarray}
Thus
\[
\#a(\textrm{NDST-II}, n)  = \#a\left(\textrm{NDST-II}, \frac{n}{2} \right) + 2\cdot \#a \left(\textrm{NDST-II}, \frac{n}{4} \right) + 2n-2
\]
Since $n=2^t$ we can obtain the second order linear difference equation with respect to $t$
\[
\#a(\textrm{NDST-II}, 2^t)  - \#a\left(\textrm{NDST-II}, 2^{t-1} \right) - 2\cdot \#a \left(\textrm{NDST-II}, 2^{t-2} \right) = 2^{t+1} - 2.
\]
Solving the above under the initial conditions $\#a \left(\textrm{NDST-II}, 2 \right) = 2$ and 
\\
$\#a \left(\textrm{NDST-II}, 4 \right) = 8$, one can obtain 
\[
\#a(\textrm{NDST-II}, 2^t)=\frac{4}{3}nt - \frac{8}{9}n - \frac{1}{9} (-1)^t+1.  
\]
Also using initial conditions $\#m \left(\textrm{NDST-II}, 2 \right) = 0$ and $\#m \left(\textrm{NDST-II}, 4 \right) = 6$, one can derive the analogous result for the number of multiplications as 
\[
\#m(\textrm{NDST-II}, 2^t)=\frac{2}{3}nt + \frac{2}{9}n + \frac{7}{9} (-1)^t-1.  
\] 
\end{proof}


\begin{corollary}\label{Lso1-41}
Let $n=2^t\:(t \geq 2)$ be given. If DST-IV algorithm ((${\bf nsin4}(n)$)) is computed using algorithms $\bf{(\ref{algo:ms2})}$ and $\bf{(\ref{algo:ms4})}$ then the arithmetic cost of computing length $n$ DST-IV algorithm is given by  
\begin{eqnarray}
\#a(\textrm{NDST-IV}, n) &=& \frac{4}{3}nt-\frac{2}{9}n+\frac{2}{9}(-1)^t,
\nonumber \\
\#m(\textrm{NDST-IV}, n) &=& \frac{2}{3}nt+\frac{14}{9}n-\frac{14}{9}(-1)^t.
\label{cs4s4}
\end{eqnarray}
\end{corollary}
\begin{proof}
The number of additions required to compute DST-IV algorithm $\bf{(\ref{algo:ms4})}$ can be found by evaluating (\ref{as24});
\begin{eqnarray}
\#a\:(\textrm{NDST-IV}, n) &=&\#a\left(\sqrt{2}\:V_n \right) + 2 \cdot \#a\left(\textrm{NDCT-II}, \frac{n}{2} \right) + \#a\left(Q_n \right)
\nonumber \\
 &=& 2 \cdot \#a\left(\textrm{NDCT-II}, \frac{n}{2} \right) + 2n - 2.
\nonumber
\end{eqnarray}
Simplifying the above with (\ref{cs2s4}) at $\frac{n}{2}$ yields
\[
\#a(\textrm{NDCT-IV}, n) = \frac{4}{3}nt-\frac{2}{9}n+\frac{2}{9}(-1)^t.
\] 
Similarly, the number of multiplications required to compute new DST-IV algorithm can be found by evaluating (\ref{as24}) with (\ref{cs2s4}) at $\frac{n}{2}$ which yields 
\[
\#m(\textrm{NDST-IV}, n)=\frac{2}{3}nt+\frac{14}{9}n-\frac{14}{9}(-1)^t.
\]
\end{proof}

The following result is trivial because the DST-III algorithm (${\bf nsin3}(n)$) was stated using the DST-II algorithm (${\bf nsin2}(n)$).


\begin{corollary}\label{Lso2-41}
Let $n=2^t\:(t \geq 2)$ be given. If DST-III algorithm (${\bf nsin3}(n)$) is computed using algorithms $\bf{(\ref{algo:ms2})}$, $\bf{(\ref{algo:ms3})}$ and $\bf{(\ref{algo:ms4})}$ then the arithmetic cost of computing length $n$ DST-III algorithm is given by  
\begin{eqnarray}
\#a(\textrm{NDST-III}, n) &=& \frac{4}{3}nt-\frac{8}{9}n-\frac{1}{9}(-1)^t+1,
\nonumber \\
\#m(\textrm{NDST-III}, n) &=& \frac{2}{3}nt+\frac{2}{9}n+\frac{7}{9}(-1)^t - 1.
\label{cs3s3}
\end{eqnarray}
\end{corollary}
\begin{remark}
Using DST-III algorithm $\bf{(\ref{algo:ms3})}$ and the arithmetic cost of DST-IV algorithm in corollary (\ref{Lso1-41}), it is possible to obtain the first order linear difference equation with respect to $t$. By solving the said equation under initial conditions $\#a(\textrm{NDST-III}, 2)=2$ and $\#m(\textrm{NDST-III}, 2)=0$ respectively, one can obtain the same results as in corollary (\ref{Lso2-41}) for the number of additions and multiplications involving in DST-III algorithm.    
\end{remark}

\begin{lemma} \label{Ls42}
Let $n=2^t\:(t \geq 2)$ be given. If DST-I algorithm (${\bf nsin1}(n -1)$) is computed using algorithms $\bf{(\ref{algo:ms1})}$, $\bf{(\ref{algo:ms2})}$, $\bf{(\ref{algo:ms3})}$ and $\bf{(\ref{algo:ms4})}$ then the arithmetic cost of DST-I algorithms of length $n-1$ is given by  
\begin{eqnarray}
\#a\:(\textrm{NDST-I}, n-1) &=& \frac{4}{3}nt-\frac{14}{9}n+\frac{1}{18}(-1)^t-t + \frac{3}{2} 
\nonumber \\
\#m\:(\textrm{NDST-I}, n-1) &=& \frac{2}{3}nt-\frac{10}{9}n-\frac{7}{18}(-1)^t + \frac{3}{2} 
\label{cs1s1}
\end{eqnarray}
\end{lemma}

\begin{proof}
Referring DST-I algorithm $\bf{(\ref{algo:ms1})}$
\begin{equation}
\#a\:(\textrm{NDST-I}, n-1)  = \#a\: \left(\textrm{NDST-I}, \frac{n}{2}-1 \right) + \#a\: \left(\textrm{NDST-III}, \frac{n}{2} \right) + \#a\: \left(\sqrt{2}\:\widehat{H}_{n-1} \right)
\label{as13}
\end{equation}
Following the structure of $\widehat{H}_{n-1}$ leads to  
\begin{equation}
\begin{matrix}
\#a \left(\sqrt{2}\:\widehat{H}_{n-1} \right) = n-2, & \#m \left(\sqrt{2}\:\widehat{H}_{n-1} \right) = 1
\label{ths13}
\end{matrix}
\end{equation}
Using arithmetic cost of DST-III (\ref{cs2s4}) at $\frac{n}{2}$ and (\ref{ths13}), we can rewrite (\ref{as13})
\[
\#a\:(\textrm{NDST-I}, n-1)  = \#a\:\left(\textrm{NDST-I}, \frac{n}{2}-1 \right) + \left( \frac{2n}{3}(t-1)- \frac{4n}{9}+\frac{1}{9}(-1)^{t}+1\right) + n-2
\]
Since $n=2^t$ the above simplifies to the first order linear difference equation with respect to $t \geq 2$
\[
\#a\:(\textrm{NDST-I}, 2^t-1)  - \#a\: \left(\textrm{NDST-I}, 2^{t-1}-1 \right) = \frac{2}{3}t\cdot2^t-\frac{1}{9}2^t+\frac{1}{9}(-1)^t-1
\]
Solving the above first order linear difference equation (with respect to $t$) using the initial condition $\#a\: \left(\textrm{NDST-I}, 1\right)= 0$, one can obtain
\[
\#a\:(\textrm{NDST-I}, 2^t-1)= \frac{4}{3}nt-\frac{14}{9}n+\frac{1}{18}(-1)^t-t+\frac{3}{2}  
\]
Also using initial condition $\#m \left(\textrm{NDST-I}, 1 \right) = 1$, one can derive the analogous result for the number of multiplications as 
\[
\#m\:(\textrm{NDST-I}, n-1) = \frac{2}{3}nt-\frac{10}{9}n-\frac{7}{18}(-1)^t + \frac{3}{2} 
\] 
\end{proof}

\subsection{Numerical illustration of the arithmetic cost of computing fast, efficient, and completely recursive DST I-IV algorithms}
\label{sebs:numcost}
The following numerical experiments are done to illustrate the number of additions and multiplications required to compute fast, efficient, and completely recursive DST I-IV algorithms having sparse, scaled orthogonal, rotational, and rotational-reflection factors. Matrices are used with the sizes from $8 \times 8$ to $4096 \times 4096$. These are implemented using MATLAB version 8.3 (R2014a).    
\\
\\
Figure (\ref{figcostADST}) and (\ref{figcostMDST}) illustrate the number of additions and multiplications required to compute DST I-IV algorithms corresponding to lemma \ref{Ls41}, corollary \ref{Lso1-41}, corollary \ref{Lso2-41}, lemma \ref{Ls42} respectively with comparison to the $n \:{\rm log}\:n$ operations.   
\\\\\
\begin{figure}[h]
\centering
\begin{subfigure}{0.4\textwidth}
\includegraphics[width=\textwidth]{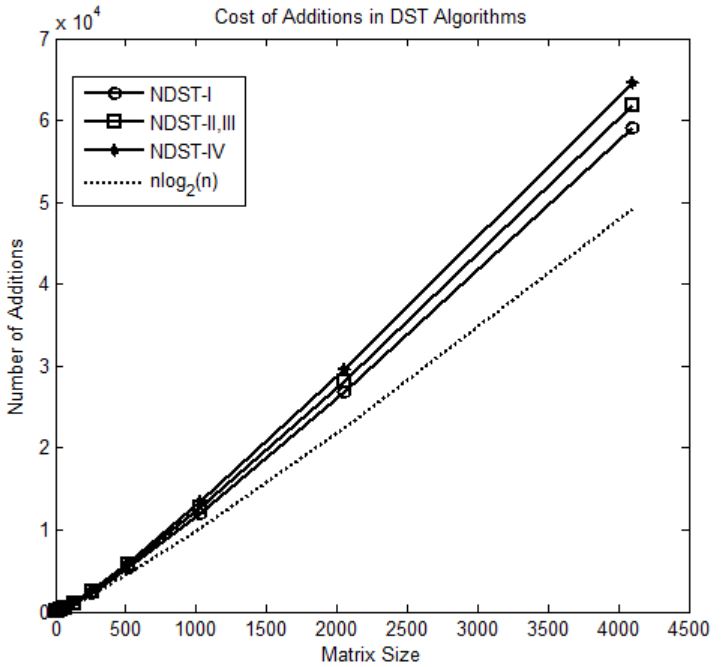}
\caption{}
\label{figcostADST}
\end{subfigure} \hspace{.3in}
\begin{subfigure}{0.4\textwidth}
\includegraphics[width=\textwidth]{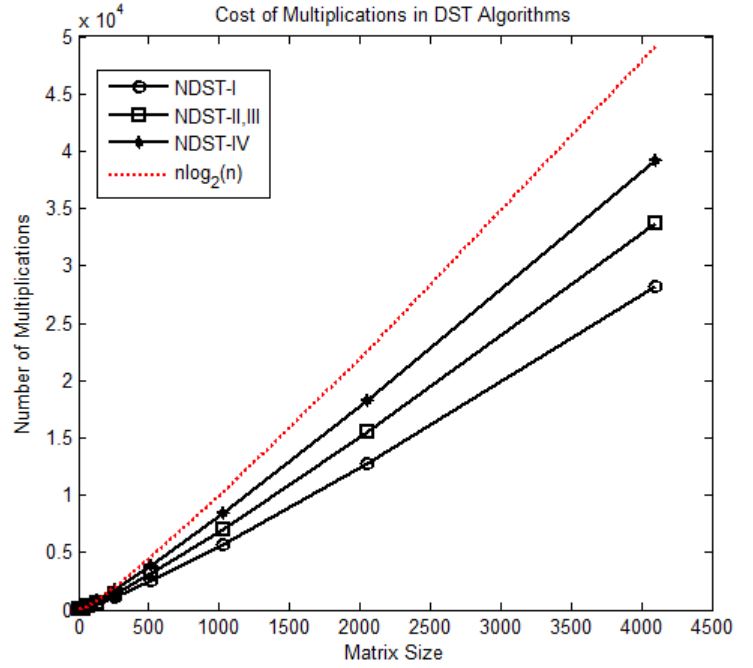}
\caption{}
\label{figcostMDST}
\end{subfigure}
\caption{(\ref{figcostADST}) Number of additions in computing DST I-IV algorithms with $n \:{\rm log}\:n$ \:(\ref{figcostMDST}) Number of multiplications in computing DST I-IV algorithms with $n \:{\rm log}\:n$ }
\label{fig:AMCS}
\end{figure}   
\section{Signal flow graphs for fast, efficient, and completely recursive DST I-IV Algorithms}
\label{sec:SFG}
In this section we use signal flow graphs to elaborate fast, efficient, and completely recursive DST I-IV algorithms having sparse, scaled orthogonal, rotation, rotation-reflection, butterfly matrices for $n=16$ and use those results to elaborate generalized $n$ points flow graphs for these DST algorithms. Note that as stated in section \ref{sec:factor}, we have developed DST I-IV algorithms to reduce the cost of multiplications. Hence, based on the cheap cost of multiplication, we can develop signal flow graphs for these DST I-Iv only by using few multipliers which is opposed to the existing DST I-IV flow graphs.      
\\\\
These signal flow graphs of DST algorithms are drawn with respect to the decimation-in-frequency having the input signal ${\bf x}$ in order and output signal ${\bf y}$ in scrambled. So for a given input signal ${\bf x}$, this section present signal flow graphs for output signal ${\bf y}=\sqrt{n}S_{n-1}^I\:{\bf x}, {\bf y}=\sqrt{n}S_{n}^{II}\:{\bf x}, {\bf y}=\sqrt{n}S_{n}^{III}\:{\bf x}$, and ${\bf y}=\sqrt{n}S_{n}^{IV}\:{\bf x}$. As shown in the flow graphs, in each graph signal flows from the left to the right. However, it is possible to convert the decimation-in-frequency DST algorithms into decimation-in-time DST algorithms applying multiplications before additions and using the identical computation complexity (same as in section \ref{sec:cost}) as in decimation-in-frequency DST algorithms.  
\\\\
In each Figure from \ref{figS1} until \ref{figGS4}, multiplication with -1 is denoted by a dotted line and notations $\epsilon:=\frac{1}{\sqrt{2}}$, $C_{i,j}:=\cos\frac{i \pi}{2^j}$, and $S_{i,j}=\sin \frac{i\pi}{2^j}$ for positive integers $i$ and $j$ are used. 


\subsection{Signal flow graphs for DST I-IV algorithms when $n=16$}
\label{ss:s16s}
Let us state the signal flow graph for DST I-IV computed via ${\bf nsin1}(n-1)$, ${\bf nsin2}(n)$, ${\bf nsin3}( n)$, ${\bf nsin4}( n)$. Here we draw the flow graphs for $n=16$ with the help of factorizations of DST I-IV algorithms as stated in section \ref{subs:FCalgo}. 
\\\\
Signal flow graphs for 15-point NDST-I ($4S^I_{15}$) and 16-point NDST II-IV (4$S^{II}_{16}$, 4$S^{III}_{16}$, and 4$S^{IV}_{16}$) algorithms are presented via Figures \ref{figS1}, \ref{figS2}, \ref{figS3}, and \ref{figS4}. 
\begin{figure}[h]
\center
\includegraphics[width=5in,height=6in]{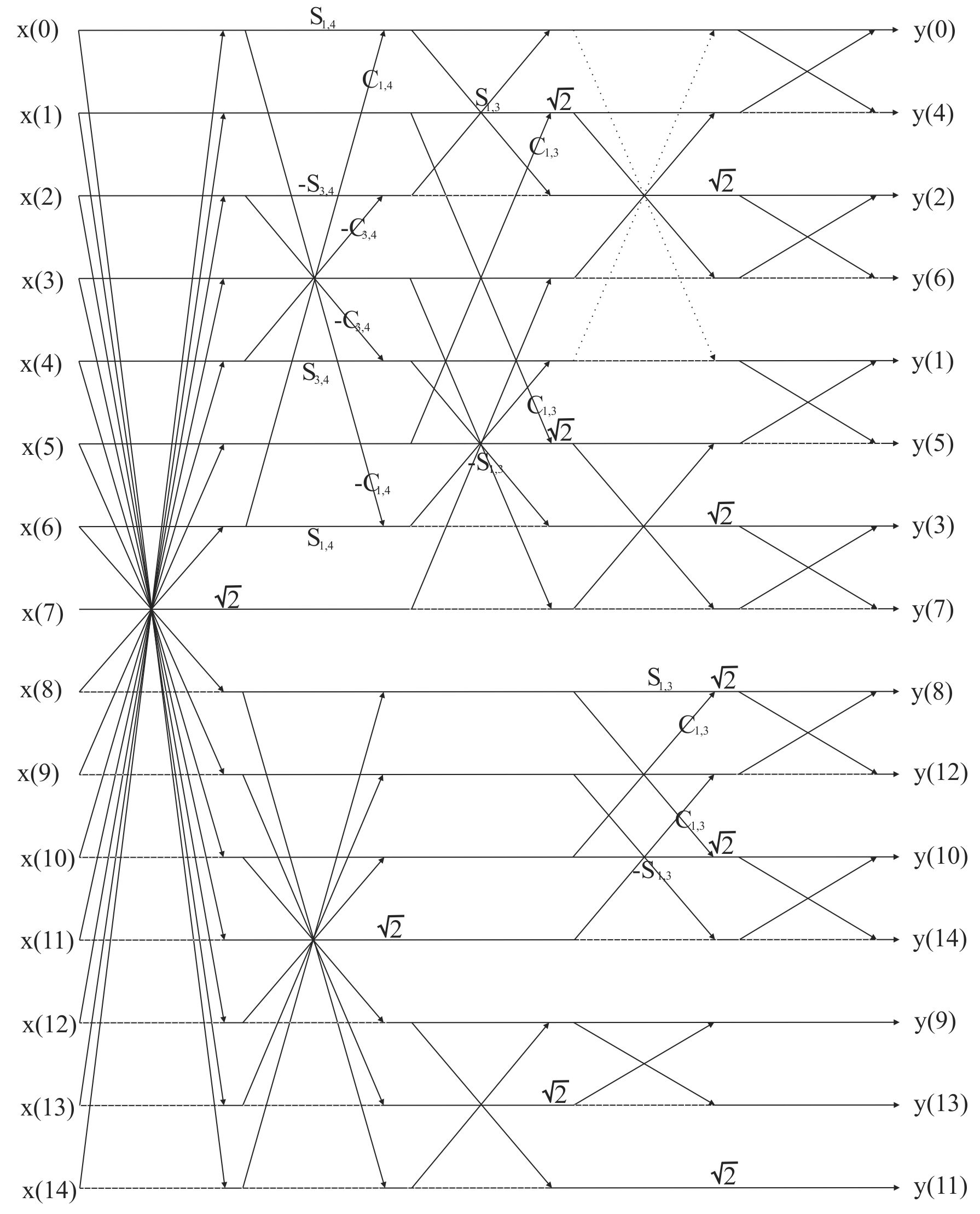}
\caption{Flow graph for 15-point NDST-I ($4S^{I}_{15}$)}
\label{figS1} 
\end{figure}
\begin{figure}[h]
\center
\includegraphics[width=5in,height=6in]{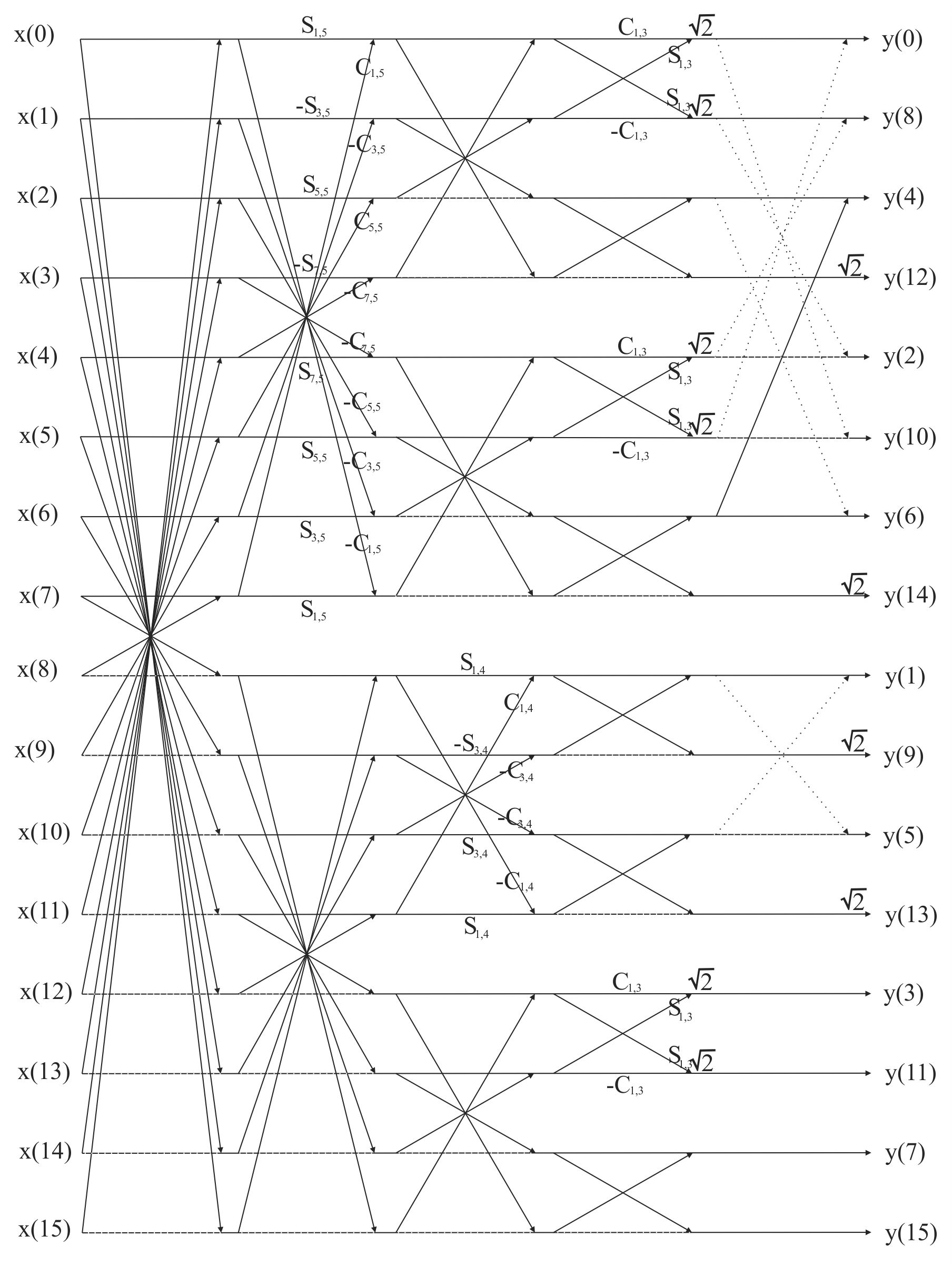}
\caption{Flow graph for 16-point NDST-II ($4S^{II}_{16}$)}
\label{figS2} 
\end{figure}
\begin{figure}[h]
\center
\includegraphics[width=5in,height=6in]{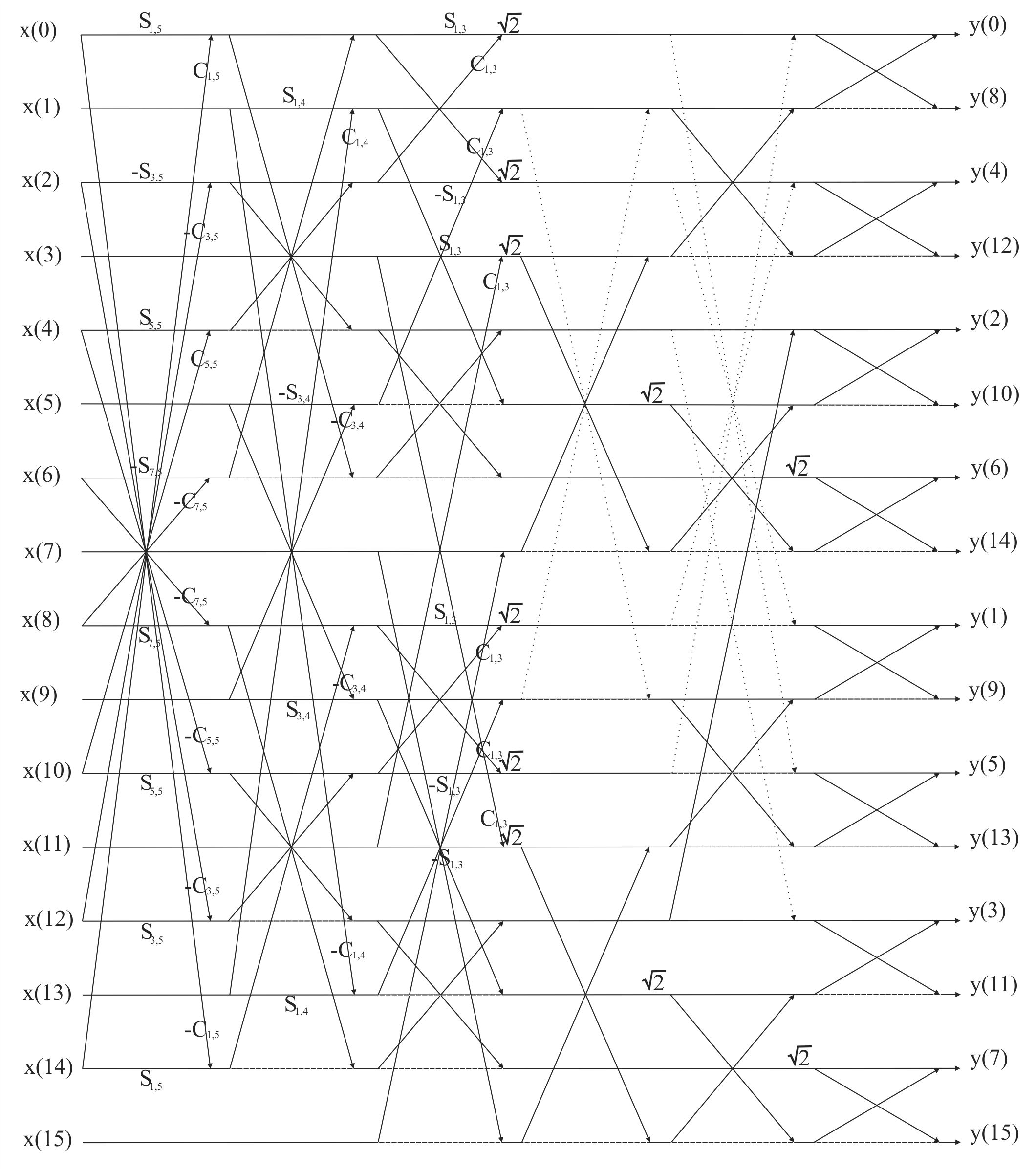}
\caption{Flow graph for 16-point NDST-III ($4S^{III}_{16}$)}
\label{figS3} 
\end{figure}
\begin{figure}[h]
\center
\includegraphics[width=5in,height=6in]{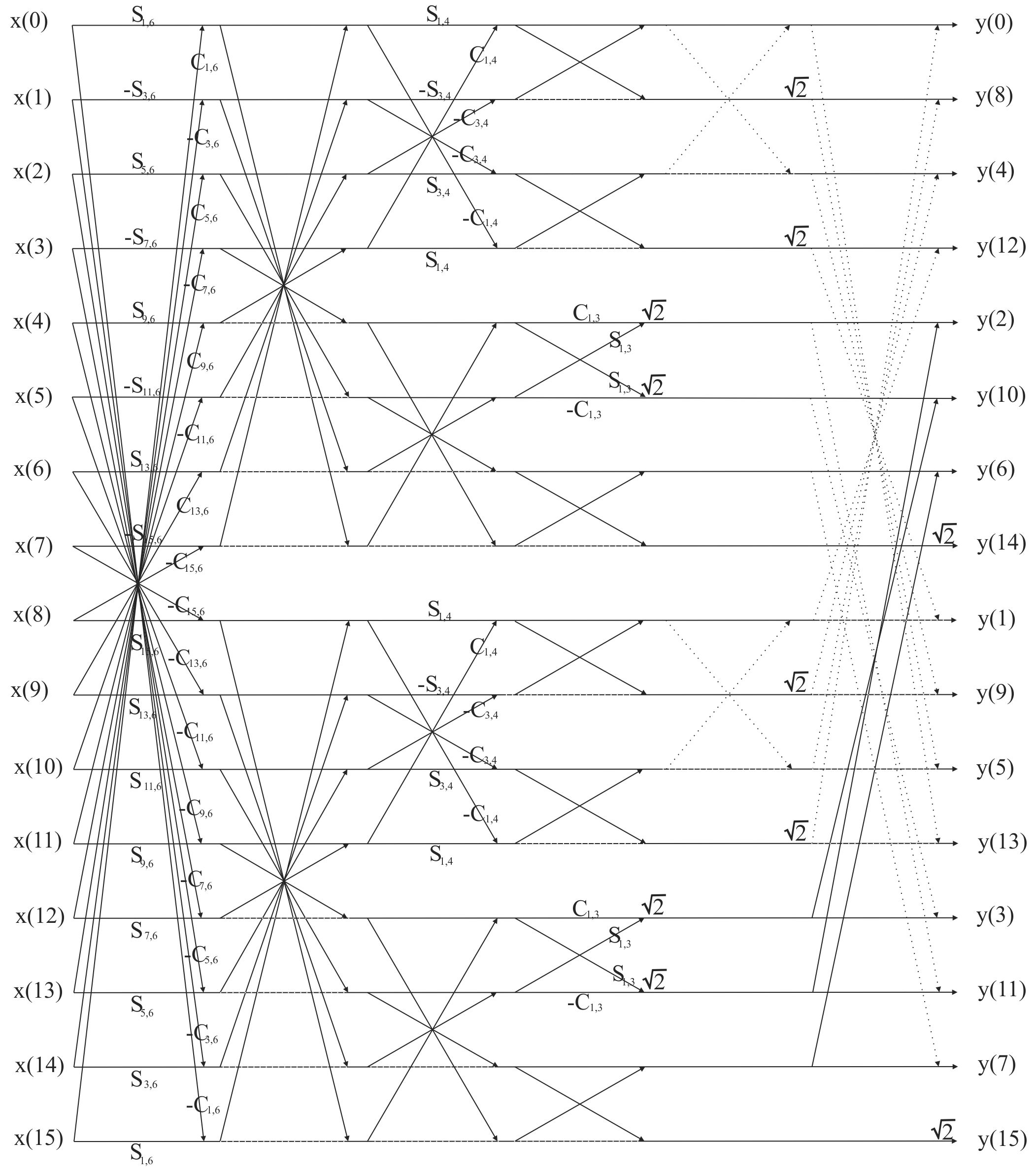}
\caption{Flow graph for 16-point NDST-IV ($4S^{IV}_{16}$)}
\label{figS4} 
\end{figure}
\\\\
The flow graphs for 16-point NDST II-IV algorithms stated via Figures \ref{figS2}, \ref{figS3}, and \ref{figS4}, the input signals ${\bf x}$ are in order and output signals ${\bf y}$ are in bit-reversed order. Thus in bit-reversed order, each output index is represented as a binary number and the indices' bits are reversed.  
\subsection{Generalized signal flow graphs for DST I-IV algorithms}
\label{ss:sgs}
Here we present generalized $(n-1)$ points signal flow graph for DST-I and $n$ points signal flow graphs for DST II-IV based on DST algorithms stated in the section \ref{subs:FCalgo} and the flow graphs drawn in the section \ref{ss:s16s}. The generalized signal flow graphs for fast and completely recursive DST I-IV algorithms can be illustrated via Figures \ref{figGS1}, \ref{figGS2}, \ref{figGS3}, and \ref{figGS4}
\begin{figure}[h]
\center
\includegraphics[width=5in,height=3.2in]{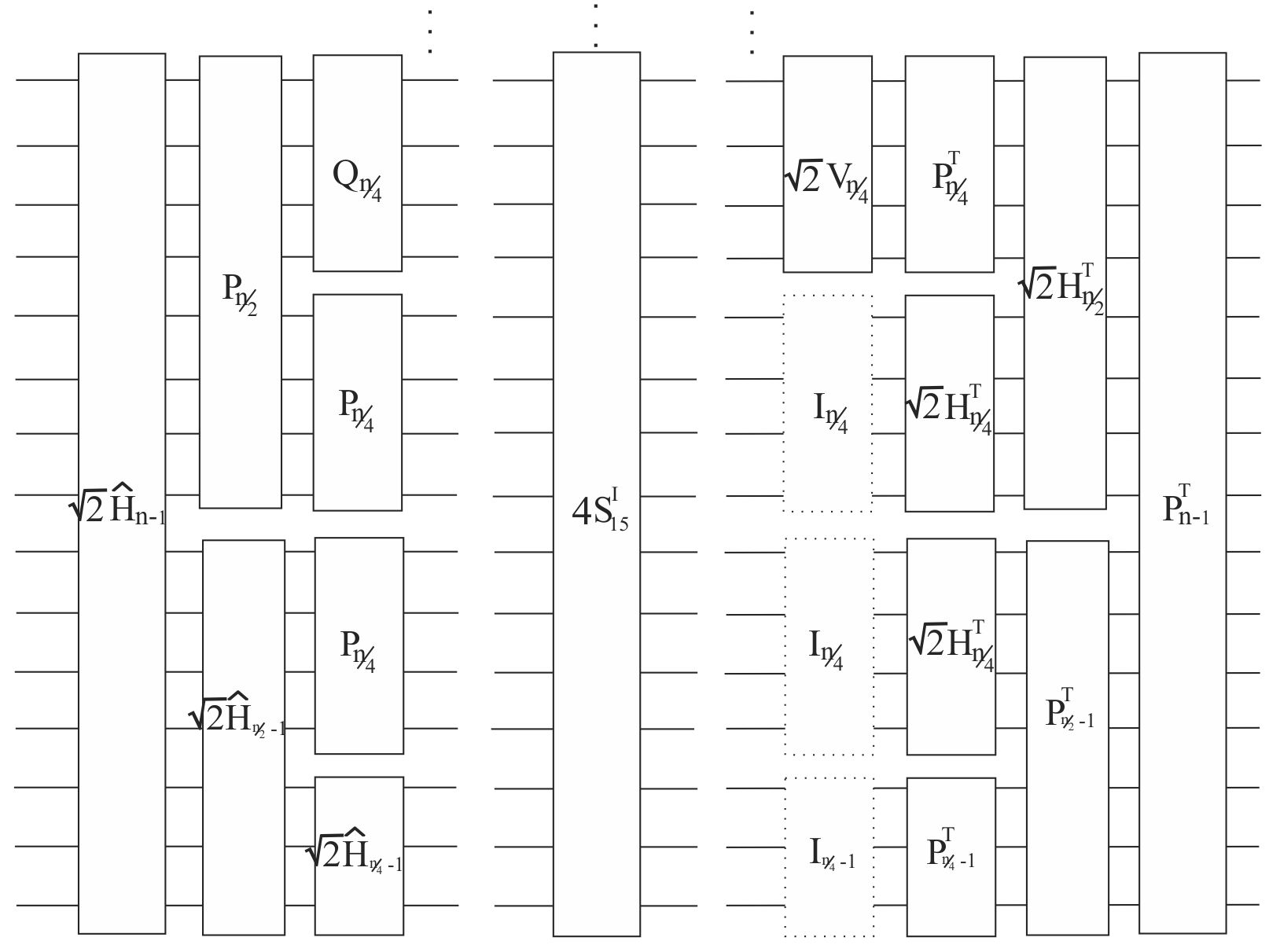}
\caption{Flow graph for $n-1$ points NDST-I ($\sqrt{n}\:S^{I}_{n-1}$)}
\label{figGS1} 
\end{figure}
\begin{figure}[h]
\center
\includegraphics[width=5in,height=3.2in]{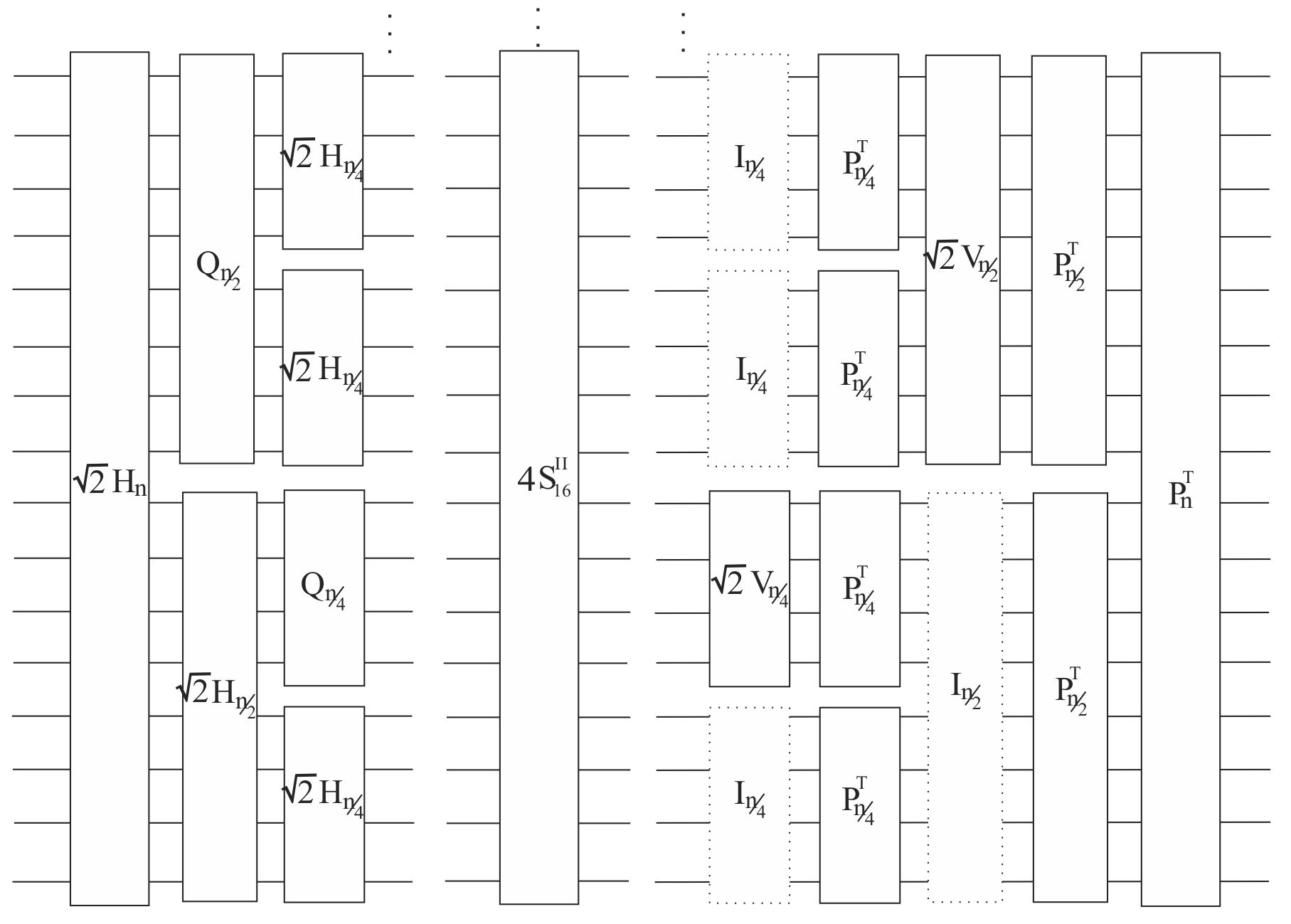}
\caption{Flow graph for $n$ points NDST-II ($\sqrt{n}\:S^{II}_{n}$)}
\label{figGS2} 
\end{figure}
\begin{figure}[h]
\center
\includegraphics[width=5in,height=3.2in]{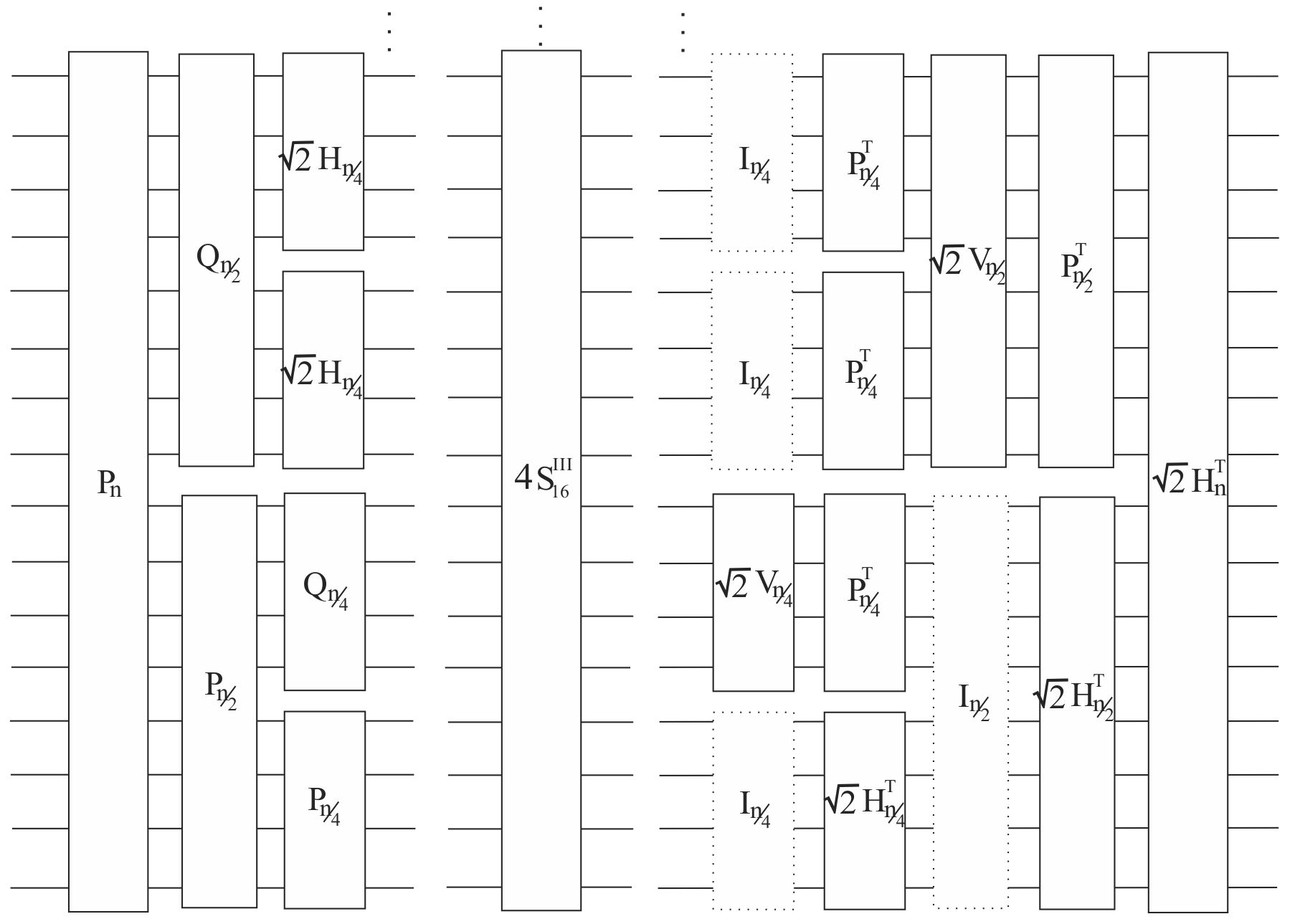}
\caption{Flow graph for $n$ points NDST-III ($\sqrt{n}\:S^{III}_{n}$)}
\label{figGS3} 
\end{figure}
\begin{figure}[h]
\center
\includegraphics[width=5in,height=3.2in]{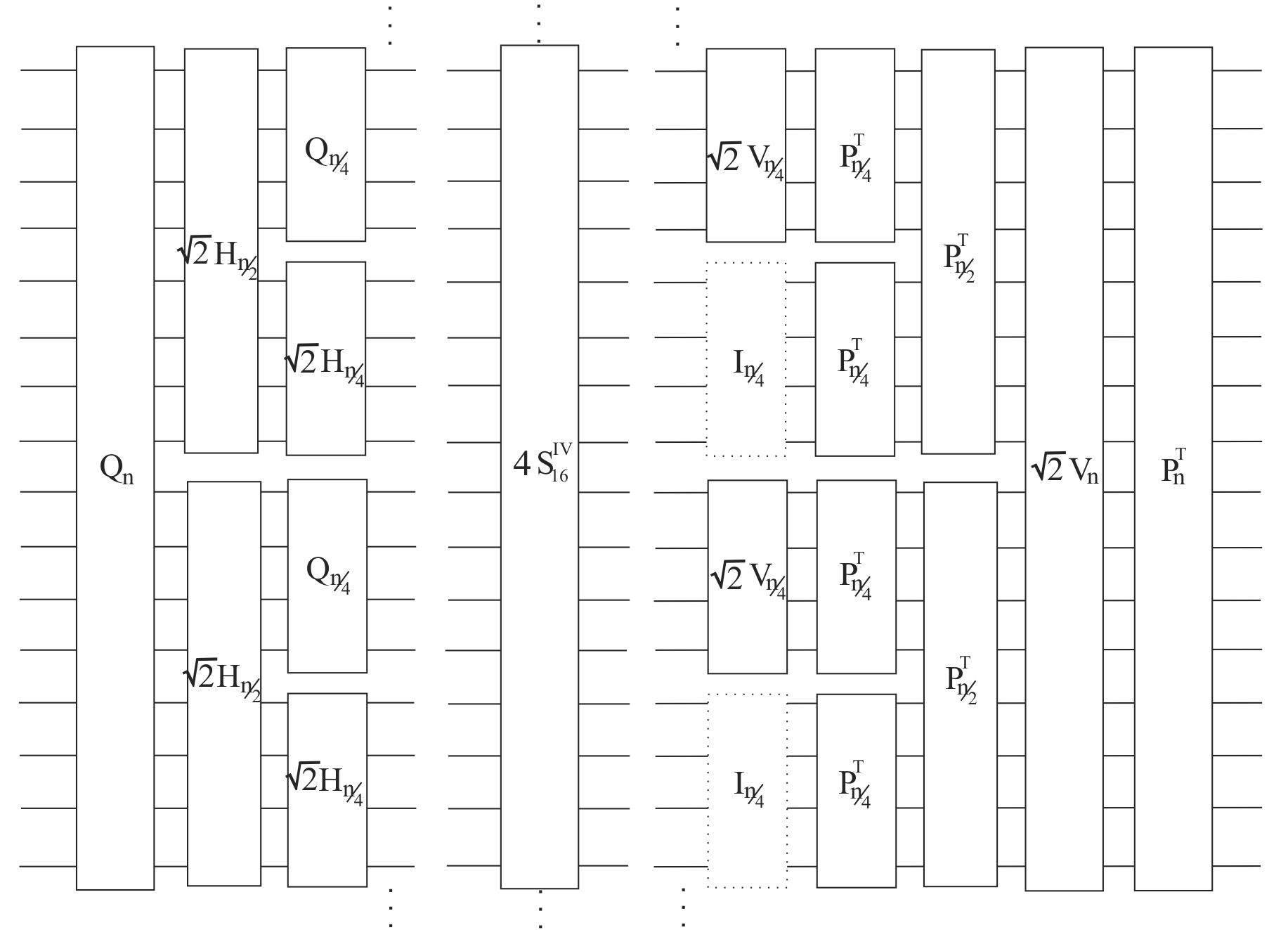}
\caption{Flow graph for $n$ points NDST-IV ($\sqrt{n}\:S^{IV}_{n}$)}
\label{figGS4} 
\end{figure}

\section{Conclusion}
In this paper, we have provided fast, efficient, and completely recursive DST I-IV algorithms, which are solely defined via DST I-IV, having sparse, scaled orthogonal, rotational, rotational-reflection, and butterfly matrices while providing the corresponding arithmetic complexity of the said algorithms. Moreover, the language of signal flow graphs is used to show the connection between factors of these DST algorithms and $(n-1)$ points DST-I flow graph and $n$ points DST II-IV flow graphs.

\newpage


\normalsize



\end{document}